\documentclass[final,12pt]{article}
\usepackage{amsfonts,color,morefloats,pslatex,a4wide}
\usepackage{amssymb,amsthm,amsmath,latexsym,pslatex}

\newtheorem{theorem}{Theorem}
\newtheorem{lemma}[theorem]{Lemma}

\newtheorem{corollary}[theorem]{Corollary}

\newtheorem{example}[theorem]{Example}

\usepackage[para]{threeparttable}

\newcommand{\ord}{{\mathrm{ord}}}

\newcommand{\tr}{{\mathrm{Tr}}}

\newcommand{\gf}{{\mathrm{GF}}}
\newcommand{\PG}{{\mathrm{PG}}}

\newcommand{\Exp}{{\mathrm{Exp}}}

\newcommand{\QR}{{\mathrm{QR}}} 
\newcommand{\QN}{{\mathrm{QN}}}

\newcommand{\Z}{\mathbb{{Z}}}

\newcommand{\m}{\mathbb{M}}

\newcommand{\C}{{\mathcal{C}}}

\newcommand{\bc}{{\mathbf{c}}}

\providecommand{\keywords}[1]
{
\textbf{\text{Keywords:}}#1
}

\makeatletter

\newcommand{\Rmnum}[1]{\expandafter\@slowromancap\romannumeral #1@}
\makeatother

\begin{document} 

\title{Several Families of Irreducible Constacyclic and Cyclic Codes\thanks{
Z. Sun's research was supported by The National Natural Science Foundation of China under Grant Number  62002093. 
X. Wang's research was supported by The National Natural Science Foundation of China under Grant Number  12001175. 
C. Ding's research was supported by The Hong Kong Research Grants Council, Proj. No. $16301522$}}

\author{Zhonghua Sun\thanks{School of Mathematics, Hefei University of Technology, Hefei, 230601, Anhui, China. Email:  sunzhonghuas@163.com}, 
\and Xiaoqiang Wang\thanks{Hubei Key Laboratory of Applied Mathematics, Faculty of Mathematics and Statistics, Hubei University, Wuhan 430062, China.  
Email: waxiqq@163.com}, 
\and Cunsheng Ding\thanks{Department of Computer Science
                           and Engineering, The Hong Kong University of Science and Technology,
Clear Water Bay, Kowloon, Hong Kong, China. Email: cding@ust.hk}
}

\maketitle

\begin{abstract} 
In this paper, several families of irreducible constacyclic codes over finite fields and their duals are studied. The weight distributions of these irreducible constacyclic codes and the parameters of their duals are settled. Several families of irreducible constacyclic codes with a few weights and several families of optimal constacyclic codes are constructed. As by-products, a family of $[2n, (n-1)/2, d \geq 2(\sqrt{n}+1)]$ irreducible cyclic codes over $\gf(q)$ and a family of $[(q-1)n, (n-1)/2, d \geq (q-1)(\sqrt{n}+1)]$ irreducible cyclic codes over $\gf(q)$ are presented, where $n$ is a prime such that $\ord_n(q)=(n-1)/2$. The results in this paper complement earlier works on irreducible constacyclic and cyclic codes over finite fields. 
\end{abstract}


\section{Introduction and motivations} 

\subsection{Constacyclic codes and cyclic codes}

Let $\gf(q)$ denote the finite field with $q$ elements, and let $\gf(q)^*$ denote the multiplicative group of $\gf(q)$. An $[n, k, d]$ code $\C$ over $\gf(q)$ is a $k$-dimensional linear subspace of $\gf(q)^n$ with minimum distance $d$. Let $\lambda \in \gf(q)^*$. A linear code $\C$ of length $n$ is said to be $\lambda$-{\it constacyclic} if $(c_0,c_1,\cdots,c_{n-1})\in \C$ implies $(\lambda c_{n-1}, c_0,c_1,\cdots,c_{n-2})\in \C$. Let $\Phi$ be the mapping from $\gf(q)^n$ to the quotient ring $\gf(q)[x]/( x^n-\lambda)$ defined by 
\[\Phi((c_0,c_1, \cdots, c_{n-1})) = c_0+c_1x+c_2x^2+ \cdots + c_{n-1}x^{n-1}.\]  
It is known that every ideal of $\gf(q)[x]/(x^n-\lambda)$ is {\it principal} and a linear code $\C \subset \gf(q)^n$ is 
$\lambda$-constacyclic if and only if $\Phi(\C)$ is an ideal of $\gf(q)[x]/(x^n-\lambda)$. Due to this, we will 
identify $\C$ with $\Phi(\C)$ for any $\lambda$-constacyclic code $\C$. 
Let $\C=( g(x))$ be a 
$\lambda$-constacyclic code over $\gf(q)$, where $g(x)$ is monic and has the smallest degree. Then $g(x)$ 
is called the {\it generator polynomial} and $h(x)=(x^n-\lambda)/g(x)$ is referred to as the {\it check polynomial} of $\C$. A $\lambda$-constacyclic code $\C$ over $\gf(q)$ is said to be {\it irreducible } if its check polynomial is irreducible over $\gf(q)$. The dual code $\C^\perp$ of $\C$ is generated by the reciprocal polynomial of the check polynomial $h(x)$ of $\C$. By definition, $1$-constacyclic codes are the classical cyclic codes. Hence, cyclic codes form a subclass of constacyclic codes. In other words, constacyclic codes are a generalisation of the classical cyclic codes. For more information on constacyclic codes, the reader is referred to \cite{Black66,CDFL15,CFLL12,DP92,DDR11,DY10,KS90,LLLM17,LQL2017,LQ2018,MC21,PD91,WSZ19,Wolfmann2008,SZW20} and the references therein.

For a linear code $\C \subseteq \gf(q)^n$, let $A_i$ denote the number of codewords with Hamming weight $i$ in $\mathcal{C}$. The {\it weight enumerator} of $\C$ is defined as $1+A_1z+\cdots+A_nz^n$. The sequence $(1,A_1,\cdots,A_n)$ is called the {\it weight distribution} of $\C$. If the number of nonzero $A_i$ in the sequence $(A_1,A_2,\cdots, A_n)$ equals $N$, then $\C$ is called an $N$-weight code. The weight distribution of a code contains important information on its error detection and correction with respect to some algorithms \cite{Klove2007}. It is well known that determining the weight distribution of a linear code is a difficult work in general, and there are a lot of references on the weight distribution of cyclic codes. For the weight distribution of irreducible cyclic codes, the reader is referred to \cite{dingyang2013} and the references therein.   

\subsection{Motivations and objectives}

 Constacyclic codes over finite fields are of theoretical importance as they are closely related to a number of areas of mathematics such as algebra, algebraic geometry, graph theory, combinatorial designs and number theory. The following remarks show that constacyclic codes have advantages over cyclic codes.
 \begin{enumerate}
  \item MDS cyclic codes over $\gf(q)$ with length $q+1$ and even dimension do not exist \cite{Geor82}, but MDS constacyclic codes over $\gf(q)$ with length $q+1$ and even dimension do exit \cite{KS90}.
  \item The Hamming code over $\gf(q)$ with parameters $\left[\frac{q^m-1}{q-1}, \frac{q^m-1}{q-1}-m, 3\right]$, here and hereafter denoted by $\text{Hamming}(q,m)$, is a perfect code. The dual of a Hamming code is the Simplex code, denoted by ${\rm Simplex}(q,m)$. The Simplex code is a one-weight code and is optimal in the sense that it meets the Griesmer bound. If $\gcd(m,q-1)=1$, the  Hamming code ${\rm Hamming}(q,m)$ is monomially-equivalent to the dual code of a cyclic code over $\gf(q)$, and the  Simplex code ${\rm Simplex}(q,m)$ is monomially-equivalent to a cyclic code over $\gf(q)$. However, this conclusion is not true when $\gcd(q-1,m)>1$. For example, the $[4,2,3]$ ternary Hamming code documented in \cite[Example 5.1.6]{HP2003} is not monomially-equivalent to a ternary cyclic code. But there is an irreducible $\lambda$-constacyclic code over $\gf(q)$ which is monomially-equivalent to the Simplex code ${\rm Simplex}(q,m)$, and its dual code is monomially-equivalent to the Hamming code $\text{Hamming}(q,m)$ \cite{FWF2017,HD2019}, where $\lambda$ is a primitive element of $\gf(q)$. 
  \item Any $[q^2+1,4,q^2-q]$ code over $\gf(q)$ is called an {\it ovoid code}. Ovoid codes are two-weight codes, and are optimal in the sense that they meet the Griesmer bound. The dual code of a $[q^2+1,4,q^2-q]$ ovoid code has parameters $[q^2+1,q^2-3,4]$ and is distance-optimal. Ovoid codes and their duals can be employed to construct $3$-designs and inversive planes \cite[Chapter 13]{dingtang2022}. Some subfield codes of ovoid codes are optimal \cite{dingheng2019}. It is known that an ovoid code corresponds to an ovoid in the projective space $\PG(3, \gf(q))$. So far the elliptic quadrics and the Tits ovoids are the only known two families of ovoids up to equivalence. The elliptic quadric in $\PG(3, \gf(2^m))$ can be constructed with an irreducible cyclic code over $\gf(2^m)$ \cite{AH2021,dingtang2022}. However, the elliptic quadric in $\PG(3, \gf(p^m))$ for odd $p$ cannot be constructed with any cyclic code over $\gf(p^m)$, but can be constructed with an irreducible constacyclic code \cite[Chapter 13]{dingtang2022}.  
 \end{enumerate}
The above remarks show that constacyclic codes can do certain things that cyclic codes cannot. Therefore, it is very interesting to study constacyclic codes.  
  
Constacyclic codes over finite fields are a subclass of the \emph{pseudo-cyclic codes} over finite fields defined in \cite[Section 8.10]{PW72}, 
and were studied under the name of pseudo-cyclic codes. Negacyclic codes over finite fields are a subclass of constacyclic codes, and were first studied by Berlekamp \cite{Black66} for correcting errors measured in the Lee metric. Therefore, the history of constacyclic codes goes back to 1966. In the past 56 years, only a few references on irreducible constacyclic codes have appeared in the literature \cite{HD2019,LQL2017,LQ2018,Wolfmann2008,SR2018,SF2020,Singh2018,ZSL18,SZW20}. It was shown in \cite{SR2018}that the weight distributions of irreducible constacyclic codes can be described in terms of the Gaussian periods of certain order $L$. However, Gaussian periods have been evaluated only for a few orders $L$. Hence, very limited results on the parameters and weight distributions of irreducible constacyclic codes over finite fields are known in the literature.   

The objectives of this paper are the following: 
\begin{enumerate} 
\item Study parameters of some families of irreducible constacyclic codes. 
\item Study parameters of the dual of these irreducible constacyclic codes.
\end{enumerate} 
The weight distributions of several families of irreducible constacyclic codes and the parameters of their duals are settled. Several infinite families of optimal constacyclic codes are presented in this paper. As by-products, a family of $[2n, (n-1)/2, d \geq 2(\sqrt{n}+1)]$ irreducible cyclic codes over $\gf(q)$ and a family of $[(q-1)n, (n-1)/2, d \geq (q-1)(\sqrt{n}+1]$ irreducible cyclic codes over $\gf(q)$ are constructed, where $n$ is a prime such that $\ord_n(q)=(n-1)/2$. 

 \subsection{The organisation of this paper}

The rest of this paper is organized as follows. In Section \ref{sec2}, we present some auxiliary results. In Section \ref{sec3}, we associate three linear codes to an irreducible constacyclic code and establish relations among the four codes. In Section \ref{sec-ding-1}, we study the parameters of irreducible constacyclic codes, and determine the weight distributions of some families of irreducible constacyclic codes. In Section \ref{sec-ding1}, we present a family of $[2n, (n-1)/2, d \geq 2(\sqrt{n}+1)]$ irreducible cyclic codes. In Section \ref{sec-ding2}, we document a family of $[(q-1)n, (n-1)/2, d \geq (q-1)( \sqrt{n}+1)]$ irreducible cyclic codes. In Section \ref{sec:7}, we present the objective of studying the two families of constacyclic codes in Sections \ref{sec-ding1} and \ref{sec-ding2}. In Section \ref{sec4}, we conclude this paper and make some concluding remarks.

\section{Preliminaries}\label{sec2}

\subsection{Cyclotomic cosets} 

Let $q$ be a prime power, $n$ be a positive integer with $\gcd(q, n)=1$, $r$ be a positive divisor of $q-1$, and let $\lambda$ be an element of $\gf(q)$ with order $r$. To deal with $\lambda$-constacyclic codes of length $n$ over $\gf(q)$, we have to study the factorization of $x^n-\lambda$ over $\gf(q)$. To this end, we need to introduce $q$-cyclotomic cosets modulo $rn$. 

Let $\Z_{rn}=\{0,1,2,\cdots,rn-1\}$ be the ring of integers modulo $rn$. For any $i \in \Z_{rn}$, the \emph{$q$-cyclotomic coset of $i$ modulo $rn$} is defined by 
\[C^{(q,rn)}_i=\{i, iq, iq^2, \cdots, iq^{\ell_i-1}\} \bmod {rn} \subseteq \Z_{rn}, \]
where $\ell_i$ is the smallest positive integer such that $i \equiv i q^{\ell_i} \pmod{rn}$, and is the \textit{size} of the $q$-cyclotomic coset. The smallest integer in $C^{(q, rn)}_i$ is called the \textit{coset leader} of $C^{(q, rn)}_i$. Let $\Gamma_{(q,rn)}$ be the set of all the coset leaders. We have then $C^{(q, rn)}_i \cap C^{(q, rn)}_j = \emptyset$ for any two distinct elements $i$ and $j$ in  $\Gamma_{(q,rn)}$, and  
 $\bigcup_{i \in  \Gamma_{(q,rn)} } C_i^{(q,rn)}=\Z_{rn}$.

Let $m=\ord_{r n}(q)$. It is easily seen that there is a primitive element $\alpha$ of $\gf(q^m)$ such that $\beta=\alpha^{(q^m-1)/rn}$ and $\beta^n=\lambda$. Then $\beta$ is a primitive $rn$-th root of unity in $\gf(q^m)$. 
 The \textit{minimal polynomial} $\m_{\beta^i}(x)$ of $\beta^i$ over $\gf(q)$ is a monic polynomial of the smallest degree over $\gf(q)$ with $\beta^i$ as a zero. We have $
\m_{\beta^i}(x)=\prod_{j \in C_i^{(q,rn)}} (x-\beta^j) \in \gf(q)[x], 
$ 
which is irreducible over $\gf(q)$. It then follows that $ x^{rn}-1=x^{rn}-\lambda^r=\prod_{i \in  \Gamma_{(q,rn)}} \m_{\beta^i}(x)$. Define 
$$\Gamma_{(q,rn)}^{(1)}=\{i: i \in \Gamma_{(q,rn)}, \, i \equiv 1 \pmod{r} \}.$$ Then $x^{n}-\lambda=\prod_{i \in  \Gamma_{(q,rn)}^{(1)}} \m_{\beta^i}(x)$.

\begin{lemma}\label{lem1} 
	Let $n$ be a positive integer with $\gcd(q,n)=1$ and let $r$ be a positive divisor of $q-1$. If $\ord_n(q)=\ell$, then $\ord_{r n}(q)=\frac{r}{\gcd(\frac{q^\ell-1}{n},r)}\ell$, which is the size $\ell_1$ of $C_1^{(q, rn)}$, and the size $\ell_i$ of each $q$-cyclotomic coset $C_i^{(q, rn)}$ is a divisor of $\ord_{rn}(q)$.
\end{lemma}

\begin{proof}
Since $r$ is a divisor of $q-1$, $\gcd(q, r)=1$. Consequently, $\gcd(q, rn)=1$. 
 It is clear that $\ord_{n}(q)$ divides $\ord_{r n}(q)$. Suppose $\ord_{rn}(q)=s \ell$. Then $s$ is the smallest positive integer such that $q^{s\ell}\equiv 1\pmod {rn}$. Clearly, $q^{s\ell}\equiv 1\pmod{rn}$ if and only if
\begin{align*}
	\frac{q^\ell-1}{n}\frac{q^{s\ell}-1}{q^\ell-1}\equiv \frac{q^\ell-1}{n} s\equiv 0 \pmod r.
\end{align*}
Hence, $s=\frac{r}{\gcd(\frac{q^\ell-1}n, r)}$, i.e., $\ord_{rn}(q)=\frac{r}{\gcd(\frac{q^\ell-1}n, r)} \ell$. The desired second conclusion is well known and its proof is thus omitted here. 
\end{proof}

\subsection{Bounds of linear codes and Pless power moments}

We first recall the following two bounds on linear codes, which will be needed in the sequel.

\begin{lemma}\label{lem2}
{\rm (Sphere Packing Bound \cite{HP2003})} Let $\C$ be an $[n, k, d]$ code over $\gf(q)$, then 
$$\sum_{i=0}^{\lfloor \frac{d-1}2\rfloor} \binom{n}{i} (q-1)^i\leq q^{n-k},$$
where $\lfloor \cdot \rfloor$ is the floor function.
\end{lemma}

The following lemma is the sphere packing bound for even minimum distances.

\begin{lemma}{\rm \cite{FWF2017}}\label{lem3}
Let $\C$ be an $[n,k,d]$ code over $\gf(q)$, where $d$ is an even integer.	Then
$$\sum_{i=0}^{\frac{d-2}2} \binom{n-1}{i} (q-1)^i\leq q^{n-1-k}.$$
\end{lemma}

An $[n, k, d]$ code over $\gf(q)$ is said to be distance-optimal if there is no $[n, k, d']$ code over $\gf(q)$ with $d'>d$. An $[n, k, d]$ code over $\gf(q)$ is said to be dimension-optimal if there is no $[n, k', d]$ code over $\gf(q)$ with $k'> k$. An $[n, k, d]$ code over $\gf(q)$ is said to be length-optimal if there is no $[n', k, d]$ code over $\gf(q)$ with $n'< n$. A linear code is said to be optimal if it is distance-optimal, or dimension-optimal or length-optimal. 
 
For an $[n, k, d]$ code $\C$ over $\gf(q)$ with weight distribution $(1, A_1,\cdots,A_n)$, we denote by $(1,A_1^{\bot},\cdots,A_n^{\bot})$ the weight distribution of its dual code. The first four Pless power moments on the two weight distributions are given as follows:
\begin{align*}
\sum_{i=0}^{n}A_i&=q^k,\\
\sum_{i=0}^{n} i A_i&=q^{k-1}[(q-1)n-A_1^{\bot}],\\	
\sum_{i=0}^n i^2 A_i&=q^{k-2}\{(q-1)^2n^2+(q-1)n-[2(q-1)n-q+2]A_1^{\bot}+2A_2^{\bot}\},\\
\sum_{i=0}^{n}i^3A_i&=q^{k-3}\{(q-1)n[(q-1)^2n^2+3(q-1)n-q+2]\\
&-[3(q-1)^2n^2-3(q-3)(q-1)n+q^2-6q+6]A_1^{\bot}\\
&+6[(q-1)n-q+2]A_2^{\bot}-6A_3^{\bot}\}.
\end{align*}

\section{An irreducible constacyclic code and its associated codes}\label{sec3}

Throughout this section, let $p$ be a prime and $q=p^s$ for a positive integer $s$. Let $n$ be a positive integer with $\gcd(n, q)=1$ and $\ell=\ord_{n}(q)$. Let $r$ be a positive divisor of $q-1$ and $\kappa=\frac{r}{\gcd(\frac{q^\ell-1}n, r)}$. Then by Lemma \ref{lem1}, $\ord_{rn}(q)=\kappa \ell$. Let $\alpha$ be a primitive element of $\gf(q^{\kappa \ell})$ and set $\beta=\alpha^{(q^{\kappa \ell}-1)/rn}$ and $\lambda=\alpha^{(q^{\kappa \ell}-1)/r}$. Then $\lambda \in \gf(q)^*$ with $\ord(\lambda)=r$ and $\beta \in \gf(q^{\kappa \ell})$ is a primitive $rn$-th root of unity such that $\beta^n=\lambda$. It should be noted that the $\lambda$ defined above ranges over all the elements of order $r$ in $\gf(q)$ when $\alpha$ ranges over all the primitive elements of $\gf(q^{\kappa \ell})$. 

Let $\C$ be the irreducible $\lambda$-constacyclic code of length $n$ over $\gf(q)$ with check polynomial $\m_{\beta}(x)$, which is the minimal polynomial over $\gf(q)$ of $\beta$. Let $\theta=\beta^{-1}$, then $\theta$ is a primitive $rn$-th root of unity. According to \cite{DY10}, \cite[Theorem 1]{SZW20} and \cite{SR2018}, we have
 \begin{equation}\label{eq1}
  	\C=\left\{ \bc(a)=\left(\tr_{ q^{\kappa \ell}/q}\left(a \theta^i\right)\right)_{i=0}^{n-1}: a \in \gf(q^{\kappa \ell}) \right\}, 
  \end{equation}
 where $\tr_{q^{\kappa \ell}/q}$ is the trace function from $\gf(q^{\kappa \ell})$ onto $\gf(q)$. The dual code $\C^{\bot}$ is the $\lambda^{-1}$-constacyclic code of length $n$ over $\gf(q)$ with generator polynomial $\m_{\beta^{-1}}(x)$.

Note that $\frac{r}{\kappa}$ divides $\frac{q^\ell-1}{n}$. We have  
$$
\left( \beta^\kappa \right)^{q^\ell-1} = \alpha^{(q^{\kappa \ell}-1) \frac{\kappa (q^\ell-1)}{rn}}=1.   
$$ 
Consequently, $\theta^\kappa \in \gf(q^\ell)$. Associated with the irreducible $\lambda$-constacyclic code $\C$ in Eq. (\ref{eq1}) are the following two codes over $\gf(q)$: 
\begin{eqnarray}\label{eqn-ding1}
 \Exp_1(\C)=\left\{ \bc_1(a)=\left(\tr_{ q^{\ell}/q}\left(a \theta^{\kappa i}\right) \right)_{i=0}^{\frac{rn}{\kappa}-1}: a\in \gf(q^\ell)\right\},
\end{eqnarray} 
and 
  $$\Exp_2(\C)=\left\{ \bc_2(a)=\left(\tr_{ q^{\ell}/q}\left(a \theta^{\kappa i}\right) \right)_{i=0}^{\frac{n}{\kappa}-1}: a\in \gf(q^\ell)\right\}.$$ 
  In the special case $\kappa=1$, $\C$ and $\Exp_2(\C)$ are identical.

   The following two conclusions follow from the trace representation of constacyclic codes (see \cite[Theorem 1]{SZW20} and \cite{SR2018}): 
 \begin{enumerate}	
 \item $\Exp_1(\C)$ is the irreducible cyclic code of length $\frac{rn}{\kappa}=\gcd(q^\ell-1, rn)$ over $\gf(q)$ with check polynomial $\m_{\beta^\kappa}(x)$, which is the minimal polynomial over $\gf(q)$ of $\beta^\kappa$.   
 \item $\Exp_2(\C)$ is the irreducible $\lambda$-constacyclic code of length $\frac{n}{\kappa}=\gcd(\frac{q^\ell-1}r,n)$ over $\gf(q)$ with check polynomial $\m_{\beta^\kappa}(x)$. 
 \end{enumerate}
  
The code $\C$ and the two codes $\Exp_1(\C)$  and  $\Exp_2(\C)$  have the following relationships.
  
  \begin{theorem}\label{thm4}
   Let notation be the same as before.
   \begin{enumerate}
   \item The cyclic code $\Exp_1(\C)=\{(\bc\parallel\lambda^{-1}\bc \parallel\cdots\parallel\lambda^{-(r-1)}\bc): \bc\in \Exp_2(\C)\}$, where $\parallel$ denotes the concatenation of vectors. 
   	\item The code $\C$ is permutation-equivalent to 
   $\underbrace{\Exp_2(\C) \oplus \cdots \oplus \Exp_2(\C)}_{\kappa ~{\rm times}}$, where $\C_1\oplus \C_2$ denotes the outer direct sum of $\C_1$ and $\C_2$, i.e., $\C_1\oplus \C_2=\{(\bc_1||\bc_2): \bc_1 \in \C_1, \bc_2 \in \C_2\}$.   
   \item The dual code $\C^{\bot}$ is permutation-equivalent to 
   $\underbrace{\Exp_2(\C)^{\bot} \oplus \cdots \oplus \Exp_2(\C)^{\bot}}_{\kappa ~{\rm times}}$.
   \end{enumerate}
   \end{theorem}
  
  \begin{proof} 1. Note that $\theta^{n}=\lambda^{-1}$. For each $ \bc_1(a) \in \Exp_1(\C)$, by definition we have 
   \begin{eqnarray*}
   \bc_1(a) &=& \left(\tr_{ q^{\ell}/q}\left(a \theta^{\kappa i}\right) \right)_{i=0}^{\frac{rn}{\kappa}-1} \\ 
   &=& (\bc_2(a)\parallel  \lambda^{-1}\bc_2(a)\parallel\cdots \parallel \lambda^{-(r-1)}\bc_2(a))
   \end{eqnarray*} 
 for any $a \in \gf(q^\ell)$, where 
   $$ 
 \bc_2(a)=\left(\tr_{ q^{\ell}/q}\left(a \theta^{\kappa i}\right) \right)_{i=0}^{\frac{n}{\kappa}-1}.    
   $$ 
  It follows that 
 $$\Exp_1(\C)=\{(\bc\parallel\lambda^{-1}\bc \parallel\cdots\parallel\lambda^{-(r-1)}\bc): \bc\in \Exp_2(\C)\}.$$  
   
2. When $\kappa=1$, $\C=\Exp_2(\C)$, the desired conclusion follows. When $\kappa>1$, 
$$ \tr_{q^\ell/q}(\tr_{q^{\kappa \ell}/q^\ell}(a \theta^{\kappa i+j }))= \tr_{q^\ell/q}(\theta^{\kappa i } \tr_{q^{\kappa \ell}/q^\ell}(a \theta^j)),$$
for $0\leq i\leq \frac{n}{\kappa}-1$ and $0\leq j\leq \kappa-1$. It is clear that $\C$ is permutation-equivalent to the code
\begin{align*}
	\mathcal{D}=\left\{\left( \bc_2(\tr_{q^{\kappa \ell}/q^{\ell}}(a)) \parallel \bc_2(\tr_{q^{\kappa \ell}/q^{\ell}}(a \theta)) \parallel \cdots \parallel \bc_2(\tr_{q^{\kappa \ell}/q^{\ell}}(a\theta^{\kappa-1})) \right) : a \in \gf(q^{\kappa \ell})   \right\}.
\end{align*}
It is easily seen that $\mathcal{D}\subseteq \underbrace{\Exp_2(\C) \oplus \cdots \oplus \Exp_2(\C)}_{\kappa ~{\rm times}}$. Note that the size of $C_1^{(q,n)}$ is equal to $\ell$, then $$\dim(\underbrace{\Exp_1(\C) \oplus \cdots \oplus \Exp_1(\C)}_{\kappa ~{\rm times}})=\kappa \ell=\dim(\C).$$
It follows that $\mathcal{D}=\underbrace{\Exp_2(\C) \oplus \cdots \oplus \Exp_2(\C)}_{\kappa ~{\rm times}}$. The desired result follows.

3. The third desired result follows directly from result 2.
  \end{proof}

By Theorem \ref{thm4}, we have the following results. 
       
\begin{theorem}\label{thm5}
Let notation be the same as before. The following assertions are equivalent.
\begin{enumerate} 
\item $\Exp_2(\C)$ is a $\left[\gcd(\frac{q^\ell-1}r, n), \ell, d\right]$ code over $\gf(q)$ with weight enumerator $W(z)$.	
\item $\Exp_1(\C)$ is a $\left[\gcd(q^\ell-1, rn), \ell, rd\right]$ code over $\gf(q)$ with weight enumerator $W(z^r)$.
\item $\C$ is an $[n, \kappa \ell, d]$ code over $\gf(q)$ with weight enumerator $W(z)^{\kappa}$.
\end{enumerate}
\end{theorem}

\begin{theorem}\label{thm6}
	Let notation be the same as before. Then $\C^{\bot}$ is an $[n, n-\kappa \ell, d^\perp]$ code over $\gf(q)$ with weight enumerator $W(z)^{\kappa}$ if and only if $\Exp_2(\C)^{\bot}$ is a $\left[\gcd(\frac{q^\ell-1}r, n), \gcd(\frac{q^\ell-1}r, n)-\ell,  d^\perp\right]$ code $\gf(q)$ with weight enumerator $W(z)$.
\end{theorem} 

The following example demonstrates the results of Theorems \ref{thm5} and \ref{thm6}. 

\begin{example} 
Let $q=4$, $n=15$ and $r=3$. Then $\ell:=\ord_n(q)=2$ , $\kappa=r=3$, and $\ord_{rn}(q)=6$. 
Let $\alpha$ be a primitive element of $\gf(q^{\kappa \ell})$ with $\alpha^{12} + \alpha^7 + \alpha^6 + \alpha^5 + \alpha^3 + \alpha + 1=0$. Then $\beta=\alpha^{(4^6-1)/3\times 15}$ and $\lambda=\alpha^{(4^6-1)/3}$ is a primitive element of $\gf(q)$. 
Then we have the following. 
\begin{itemize}
\item $\Exp_2(\C)$ is a $\left[5, 2, 4\right]$ code over $\gf(4)$ with weight enumerator $W(z)=1+15z^4$, 
and $\Exp_2(\C)^\perp$ is a $\left[5, 3, 3\right]$ code over $\gf(4)$ with weight enumerator $W^\perp(z)=1+30z^3+15z^4+18z^5$. 
\item $\C$ is a $[15, 6, 4]$ code over $\gf(4)$ with weight enumerator $W(z)^{3}$, 
and $\C^\perp$ is a $[15, 9, 3]$ code over $\gf(4)$ with weight enumerator $(W^\perp(z))^{3}$. 
\item $\Exp_1(\C)$ is a $\left[15, 2, 12\right]$ code over $\gf(4)$ with weight enumerator $W(z^3)=1+15z^{12}$.
\end{itemize}
\end{example} 

The following theorem gives an estimate of the dual distance of the irreducible constacyclic code in Eq. (\ref{eq1}).

 \begin{theorem}\label{thm7}
 Let $\C$ be the irreducible $\lambda$-constacyclic code in Eq. (\ref{eq1}). Then $\C^{\bot}$ is a $\lambda^{-1}$-constacyclic code over $\gf(q)$ with parameters $\left[n, n- \frac{r}{\gcd(\frac{q^\ell-1}n,r)}\ell, d^{\bot} \geq 2\right]$, where $r=\ord(\lambda)$ 
 and $\ell=\ord_n(q)$. Furthermore, if $\ell \geq 2$, then the following hold:  
 \begin{enumerate}
 	\item $d^{\bot}= 2$ if and only if $\gcd(\frac{q-1}r,n)>1$.
 	\item If $n\mid (\frac{q^\ell-1}{q-1})$, $n> 2(\frac{q^{\ell/2}-1}{q-1})$, and $\gcd(\frac{q-1}r,n)=1$, then $3\leq d^{\bot}\leq 4$. 
 	\item If $n\mid (\frac{q^\ell-1}{q-1})$, $n>\frac{q^{\ell-1}-1}{q-1}+2$, and $\gcd(\frac{q-1}r,n)=1$, then  $\C^{\bot}$ is an $[n, n-\ell, 3]$ code over $\gf(q)$, which is both distance-optimal and dimension-optimal. 
 \end{enumerate}
 \end{theorem}
 
 \begin{proof}
 Note that $\C$ is an $[n, \kappa \ell]$ $\lambda$-constacyclic code. It is well known that $\C^{\bot}$ is an $[n, n-\kappa \ell]$ $\lambda^{-1}$-constacyclic code. It is clear that $d^{\bot}\geq 2$. Obviously, $d^{\bot}=2$ if and only if there is an integer $1\leq i \leq n-1$ such that $\theta^i \in \gf(q)^*$, i.e., $(q-1)i\equiv 0 \pmod {rn}$ has a solution $1\leq i \leq n-1$, which is equivalent to $\gcd(\frac{q-1}r, n)>1$.
 
Assume now that $n \mid (\frac{q^\ell-1}{q-1})$ and $\gcd(\frac{q-1}r,n)=1$. Then $\kappa=1$ and $\C^{\bot}$ is an $[n,n-\ell, d^\perp]$ code over $\gf(q)$ with $d^\perp \geq 3$. 
 If $n>2(\frac{q^{\ell/2}-1}{q-1})$, by the Sphere Packing bound, the minimum distance of an $[n, n-\ell]$ code over $\gf(q)$ is at most $4$. Hence, $d^{\bot}\leq 4$. 
If  $n> \frac{q^{\ell-1}-1}{q-1}+2$, there is no $[n, n-\ell. 4]$ code over $\gf(q)$ by the bound in Lemma \ref{lem3}, which means that $\C^{\bot}$ is distance-optimal, and there is no $[n, n-\ell+1,3]$ code over $\gf(q)$ by the Sphere Packing bound, which means that $\C^{\bot}$ is dimension-optimal. 
 \end{proof}
 
 The following is a corollary of Theorem \ref{thm7}.  
  
 \begin{corollary}\label{cormarch231}
 	Let $n=\frac{q^m-1}{(q-1)e}$, where $m$ is a positive integer, $e$ is a positive divisor of $\frac{q^m-1}{q-1}$ and $e\leq q-1$. 
 	Let $\gcd(\frac{q-1}r,n)=1$ and let $\C$ be the irreducible $\lambda$-constacyclic code in Eq. (\ref{eq1}). Then $\C^{\bot}$ is an $[n, n-m, 3]$ code over $\gf(q)$ and is both distance-optimal and dimension-optimal.
 \end{corollary}
 
The following example demonstrates the dual code $\C^\perp$ in Corollary \ref{cormarch231}.  
 
 \begin{example}
   Let $m=4$, $q=5$, and $r=4$. Then we have the following examples of the code $\C^\perp$ in Corollary \ref{cormarch231}.
   \begin{enumerate}
   	\item When $e=1$ and $\alpha$ is the primitive element of $\gf(5^4)$ with $\alpha^4 + 4\alpha^2 + 4\alpha + 2=0$, 
   	we have $n=156$  and $\lambda=2$. The dual code of the irreducible $2$-constacyclic code $\C$ over $\gf(5)$ 
   	in Corollary \ref{cormarch231}  has  parameters $[156,152,3]$.
   	\item When $e=2$ and $\alpha$ is the primitive element of $\gf(5^4)$ with $\alpha^4 + 4\alpha^2 + 4\alpha + 2=0$, 
   	we have $n=78$  and $\lambda=2$. The dual code of the irreducible $2$-constacyclic code $\C$ over $\gf(5)$ 
   	in Corollary \ref{cormarch231}  has  parameters $[78,74,3]$.  
   	   	\item When $e=3$ and $\alpha$ is the primitive element of $\gf(5^4)$ with $\alpha^4 + 4\alpha^2 + 4\alpha + 2=0$, 
   	we have $n=52$  and $\lambda=2$. The dual code of the irreducible $2$-constacyclic code $\C$ over $\gf(5)$ 
   	in Corollary \ref{cormarch231}  has  parameters $[52,48,3]$. 
   	   	\item When $e=4$ and $\alpha$ is the primitive element of $\gf(5^4)$ with $\alpha^4 + 4\alpha^2 + 4\alpha + 2=0$, 
   	we have $n=39$  and $\lambda=2$. The dual code of the irreducible $2$-constacyclic code $\C$ over $\gf(5)$ in Corollary \ref{cormarch231} has  parameters $[39,35,3]$. 	
   \end{enumerate} 
All the codes $\C^\perp$ in this example are both dimension-optimal and distance-optimal. 
 \end{example}
 
We now introduce another irreducible cyclic code associated to the irreducible constacyclic code $\C$ in Eq. (\ref{eq1}). Note that $\beta^{r}$ is an $n$-th primitive root of unity in $\gf(q^{\kappa \ell})$, where $\ell=\ord_n(q)$ and 
$$
\kappa=\frac{r}{\gcd((q^\ell -1)/n, r)}. 
$$ 
Let $\Exp_3(\C)$ be the irreducible cyclic code of length $n$ over $\gf(q)$ with check polynomial $\m_{\beta^r}(x)$. 
It is easily seen that 
$$
\m_{\beta^r}(x)=\prod_{i=0}^{\ell -1} (x-\beta^{ri}). 
$$
Consequently, $\Exp_3(\C)$ is an $[n, \ell]$ cyclic code over $\gf(q)$, while $\C$ is an $[n, \kappa \ell]$ irreducible constacyclic code over $\gf(q)$. 
The four codes $\C$, $\Exp_1(\C)$, $\Exp_2(\C)$ and $\Exp_3(\C)$ are related by definition. In a special case, we have the following. 

\begin{theorem}\label{thm-bridge1} 
If $\gcd(r, n)=1$, then  $\C$ and $\Exp_3(\C)$ have length $n$ and dimension $\ell$ and are permutation-equivalent. 
\end{theorem} 

\begin{proof} 
Let $\gcd(r, n)=1$. Then we have $\kappa=1$. By Theorem \ref{thm5}, 
$\C$  has length $n$ and dimension $\ell$. Let $\theta=\beta^{-1}$. 
By definition, the trace representation of $\Exp_3(\C)$ is given by 
\begin{eqnarray*} 
 \Exp_3(\C)=\left\{ \bc_3(a)=\left(\tr_{ q^{\ell}/q}\left(a \theta^{ri}\right) \right)_{i=0}^{n-1}: a\in \gf(q^\ell)\right\}. 
\end{eqnarray*} 
Since  $\gcd(r, n)=1$, $\kappa=1$ and the trace representation of $\C$ in Eq. (\ref{eq1}) becomes 
\begin{eqnarray*} 
 \C=\left\{ \bc(a)=\left(\tr_{ q^{\ell}/q}\left(a \theta^{i}\right) \right)_{i=0}^{n-1}: a\in \gf(q^\ell)\right\}.  
\end{eqnarray*} 
It is easily seen that the coordinate permutation 
$i \mapsto ir \bmod n$ sends $\C$ to $\Exp_3(\C)$.  
Therefore, $\C$ and $\Exp_3(\C)$ are permutation-equivalent. 
\end{proof} 

When $\gcd(r, n)=1$, the relations among the  weight enumerators and parameters of the four codes 
$\C$, $\Exp_1(\C)$, $\Exp_2(\C)$ and $\Exp_3(\C)$ are clearly known by combining Theorems \ref{thm5} and \ref{thm-bridge1}. 
It is very interesting to note that the irreducible $\lambda$-constacyclic code $\C$ in Eq. (\ref{eq1}) is permutation-equivalent 
to the irreducible cyclic code $\Exp_3(\C)$ under the special condition $\gcd(r, n)=1$. This is a special result only for 
this special code $\C$ in Eq. (\ref{eq1}) under this special condition. Theorem \ref{thm-bridge1} will be used to study 
two families of irreducible constacyclic codes in Sections \ref{sec-ding1} and \ref{sec-ding2}. In fact, all the known 
results about the irreducible cyclic code $\Exp_3(\C)$ surveyed in \cite{dingyang2013} can be translated into similar 
results about the irreducible $\lambda$-constacyclic code $\C$ under the condition that $\gcd(r,n)=1$ and $r=\ord(\lambda)$.     

We remark that the three codes $\Exp_1(\C)$,  $\Exp_2(\C)$ and $\Exp_3(\C)$ are associated only to the irreducible $\lambda$-constacyclic code $\C$ in Eq. (\ref{eq1}), although similar codes may be associated to a general $\lambda$-constacyclic code. 
This irreducible $\lambda$-constacyclic code $\C$ in Eq. (\ref{eq1}) is very special in the following senses:   
\begin{itemize}
\item Its dimension is known to be $\kappa \ell$, while the dimension of other irreducible $\lambda$-constacyclic codes is knwon to be a divisor of  $\kappa \ell$ only. 
\item Its trace representation is very simple and cannot be reduced to a trace representation over a proper subfield of $\gf(q^{\kappa \ell})$.  
\end{itemize} 
It should be informed that this paper studies only this special irreducible $\lambda$-constacyclic code $\C$ in Eq. (\ref{eq1})
and its associated codes. 

Let $\C^{(t)}$ denote the $\lambda$-constacyclic code of length $n$ over $\gf(q)$ with check polynomial $\m_{\beta^t}(x)$. 
It is easily seen from the trace representations of $\C$ and $\C^{(t)}$ that the two codes are permutation-equivalent if 
$\gcd(t, n)=1$ and $t \equiv 1 \pmod{r}$ \cite{SR2018}. Almost all the results about $\C$ presented in this paper are also valid for all the $\lambda$-constacyclic codes $\C^{(t)}$ with  $\gcd(t, n)=1$  and $t \equiv 1 \pmod{r}$.

\section{Parameters of the code $\C$ in (\ref{eq1}) in the case $\gcd(\frac{q-1}r,n)=1$}\label{sec-ding-1}

In this section, we consider only the case that $\gcd(\frac{q-1}r,n)=1$. Recall that $\ell=\ord_n(q)$ and $\gcd(q,n)=1$. 
In this case, we have 
$$\kappa=\frac{r}{\gcd(\frac{q^\ell-1}n,r)}=\frac{q-1}{\gcd(\frac{q^\ell-1}n, q-1)}. $$
It follows that $\C$ is an $\left[n,\frac{(q-1)\ell}{\gcd(\frac{q^\ell-1}n, q-1)}\right ]$ code over $\gf(q)$ in this case. Theorem \ref{thm5} shows that determining the parameters of the irreducible constacyclic code $\C$ is equivalent to determining the parameters of the irreducible cyclic code $\Exp_1(\C)$. Recall that the irreducible cyclic code $\Exp_1(\C)$ has length $\gcd(q^\ell-1, rn) =r\gcd(\frac{q^\ell-1}{q-1},n)$ and check polynomial $\m_{\beta^{\kappa}}(x)$. Let 
     \begin{equation}\label{eqn-oure}
     	e=\gcd(\frac{q^\ell-1}{q-1}, \frac{q^\ell-1}{\gcd(q^\ell-1,rn)})=\frac{(q^\ell-1) \gcd(q-1,rn)}{(q-1)\gcd(q^\ell-1,rn)}=\frac{q^\ell-1}{\gcd(q^\ell-1,(q-1)n)}.
     \end{equation} 
      
  The following lemma follows directly from the results in \cite{dingyang2013}. We will use it later to prove a main result of this paper. 
        
  \begin{lemma}\label{lem8}
  Let notation be as before, and let $r$ be a positive divisor of $q-1$. 
  	  	\begin{enumerate}
  	\item If $n=r(\frac{q^m-1}{q-1})$ for some integer $m \geq 2$, then  $e=1$ and  $\Exp_1(\C)$ is an $\left[r(\frac{q^m-1}{q-1}), m, r q^{m-1}\right]$ one-weight cyclic code.	
  	\item If $n=r(\frac{q^m-1}{2(q-1)})$ for some even integer $m \geq 2$ and $\gcd(r,2)=1$, then $e=2$ and $\Exp_1(\C)$ is an $\left[r(\frac{q^m-1}{2(q-1)}), m, \frac{r (q^{m-1}-q^{\frac{m-2}2})}2 \right]$ two-weight cyclic code with weight enumerator
  	$$1+ \left(\frac{q^m-1}2\right) z^{\frac{ r(q^{m-1}-q^{\frac{m-2}2})}2 }+\left(\frac{q^m-1}2\right) z^{\frac{r (q^{m-1}+q^{\frac{m-2}2})}2 }.$$
  	\item If $n=r(\frac{q^m-1}{3(q-1)})$ for some odd integer $m \geq 3$, $sm \equiv 0 \pmod{3}$, $\gcd(r, 3)=1$ and $p\equiv 1 \pmod {3}$, then $e=3$ 
  	and $\Exp_1(\C)$ is an $\left[r(\frac{q^m-1}{3(q-1)}), m \right]$ three-weight cyclic code with weight enumerator
  \begin{align*}
  	1&+ \left(\frac{q^m-1}3\right) z^{\frac{r (q^{m-1}-c_1q^{\frac{m-3}2})}3 }+\left(\frac{q^m-1}3\right) z^{\frac{r (q^{m-1}+\frac{1}2(c_1+9d_1)q^{\frac{m-3}2})}3 }\\
  	&+\left(\frac{q^m-1}3\right) z^{\frac{r (q^{m-1}+\frac{1}2(c_1-9d_1)q^{\frac{m-3}2})}3 }
  \end{align*}
  	where $c_1$ and $d_1$ are given by $4q^{m/3}=c_1^2+27d_1^2$, $c_1\equiv 1 \pmod {3}$ and $\gcd(c_1,p)=1$.
  	
  	\item If $n=r(\frac{q^m-1}{4(q-1)})$ for some even integer $m \geq 4$ with $sm \equiv 0 \pmod{4}$, $\gcd(r, 2)=1$, and $p\equiv1 \pmod {4}$, then $e=4$ and $\Exp_1(\C)$ is an $\left[r(\frac{q^m-1}{4(q-1)}),m \right]$ cyclic code with weight enumerator 
  	\begin{align*}
  	1+& \left(\frac{q^m-1}4\right) z^{\frac{r(q^{m-1}+q^{\frac{m-2}2}+2c_1q^{\frac{m-4}4})}{4} }+\left(\frac{q^m-1}4\right) z^{\frac{r(q^{m-1}+q^{\frac{m-2}2}-2c_1q^{\frac{m-4}4})}{4} }\\
  &+ \left(\frac{q^m-1}4\right) z^{\frac{r(q^{m-1}-q^{\frac{m-2}2}+4d_1q^{\frac{m-4}4})}{4} }+\left(\frac{q^m-1}4\right) z^{\frac{r(q^{m-1}-q^{\frac{m-2}2}-4d_1q^{\frac{m-4}4})}{4} }
  \end{align*}
  	where $c_1$ and $d_1$ are given by $q^{m/2}=c_1^2+4d_1^2$, $c_1\equiv 1 \pmod {4}$ and $\gcd(c_1,p)=1$.
  
  	\item Let $m \geq 2$ be an even integer and $e>2$ be a divisor of $(q^m-1)/(q-1)$ such that $\gcd(e, r)=1$. Let 
  	$n=r(\frac{q^m-1}{(q-1)e})$.  
  	 If  $p^j\equiv -1 \pmod {e}$ for a positive integer $j$. Assume that $j$ is the smallest positive integer such that $p^j\equiv -1 \pmod {e}$. Define $\gamma=\frac{s m}{2j}$.
  	\begin{enumerate}
  	 \item If $\gamma$, $p$ and $\frac{p^j+1}e$ are all odd, then $\Exp_1(\C)$ is an $\left[r(\frac{q^m-1}{(q-1)e}), m \right]$ two-weight cyclic code with weight enumerator  
  	 $$1+\left(\frac{q^m-1}{e}\right)z^{\frac{r(q^{m-1}-(e-1)q^{\frac{m-2}2})}{e}} +\left(q^m-1-\frac{q^m-1}e\right)z^{\frac{r(q^{m-1}+q^{\frac{m-2}2})}{e}},$$
  	 provided that $e<q^{m/2}+1$.
  	 \item In all other cases, then $\Exp_1(\C)$ is an $\left[r(\frac{q^m-1}{(q-1)e}), m \right]$ two-weight cyclic code with weight enumerator  
  	 $$1+\left(\frac{q^m-1}{e}\right)z^{\frac{r(q^{m-1}+(-1)^\gamma(e-1)q^{\frac{m-2}2})}{e}} +\left(q^m-1-\frac{q^m-1}{e}\right)z^{\frac{r(q^{m-1}-(-1)^{\gamma}q^{\frac{m-2}2})}{e}},$$
  	 provided that $q^{m/2}+(-1)^{\gamma}(e-1)>0$.
  	\end{enumerate}
  	\end{enumerate}
  \end{lemma}

One of the main results of this paper is documented in the following theorem.   
 
  \begin{theorem}\label{thm9}
  Let $n=u(\frac{q^m-1}{(q-1)e})$, where $m\geq2$, $e$ is a positive divisor of $(q^m-1)/(q-1)$ and $e \leq (q^{m/2}+1)/2$, $u$ is a positive divisor of $q-1$ and $\gcd(e, u)=1$.  Let $\C$ be the irreducible $\lambda$-constacyclic code in Eq. (\ref{eq1}). 
  \begin{enumerate}
  	\item If $e=1$ and $\gcd(\frac{q-1}r, n)=1$, then $\C$ is an $\left[n, um, q^{m-1}\right]$ code over $\gf(q)$ with weight enumerator $$\left[1+(q^m-1)z^{q^{m-1}}\right]^u.$$
  The dual code $\C^{\bot}$ has parameters $[n, n-um, 3]$.  
  \item If $m$ is even, $e=2<q$ and $\gcd(\frac{q-1}r, n)=1$, then $\C$ is an $\left[n, um, \frac{q^{m-1}-q^{\frac{m-2}2}}2 \right]$ code over $\gf(q)$ with weight enumerator
  	$$\left[1+ \left(\frac{q^m-1}2\right) z^{\frac{ q^{m-1}-q^{\frac{m-2}2}}2 }+\left(\frac{q^m-1}2\right) z^{\frac{q^{m-1}+q^{\frac{m-2}2}}2 }\right]^{u}.$$
   The dual code $\C^{\bot}$ has parameters $[n, n-um, 3]$.  	
   
   \item If $e=3<q$, $m \geq 3$ is odd, $sm \equiv 0 \pmod{3}$,  $p\equiv 1 \pmod {3}$ and  $\gcd(\frac{q-1}r, n)=1$, then $\C$ is an $[n, um ]$ code over $\gf(q)$ with weight enumerator
  \begin{align*}
  	\left[1+ \left(\frac{q^m-1}3\right) \left( z^{\frac{q^{m-1}-c_1q^{\frac{m-3}2}}3}+
  	            z^{\frac{q^{m-1}+\frac{1}2(c_1+9d_1)q^{\frac{m-3}2}}3 }+
  	           z^{\frac{q^{m-1}+\frac{1}2(c_1-9d_1)q^{\frac{m-3}2}}3 } \right) \right]^{u}
  \end{align*}
  	where $c_1$ and $d_1$ are given by $4q^{m/3}=c_1^2+27d_1^2$, $c_1 \equiv 1 \pmod {3}$ and $\gcd(c_1,p)=1$.
 The dual code $\C^{\bot}$ has parameters $[n, n-um, 3]$.  
 \item If $e=4<q$, $m \geq 4$ is even, $sm \equiv 0 \pmod{4}$, $p \equiv 1 \pmod {4}$ and  $\gcd(\frac{q-1}r, n)=1$, 
 then $\C$ is an $[n, um ]$ code 
 over $\gf(q)$ with weight enumerator 
  	\begin{align*}
  \left[1+\left(\frac{q^m-1}4\right) z^{\frac{q^{m-1}+q^{\frac{m-2}2}+2c_1q^{\frac{m-4}4}}{4} }+\left(\frac{q^m-1}4\right) z^{\frac{q^{m-1}+q^{\frac{m-2}2}-2c_1q^{\frac{m-4}4}}{4} }\right.\\
 \left.+ \left(\frac{q^m-1}4\right) z^{\frac{q^{m-1}-q^{\frac{m-2}2}+4d_1q^{\frac{m-4}4}}{4} }+\left(\frac{q^m-1}4\right) z^{\frac{q^{m-1}-q^{\frac{m-2}2}-4d_1q^{\frac{m-4}4}}{4} }\right]^u,
  \end{align*}
  	where $c_1$ and $d_1$ are given by $q^{m/2}=c_1^2+4d_1^2$, $c_1\equiv 1 \pmod {4}$ and $\gcd(c_1,p)=1$.
  The dual code $\C^{\bot}$ has parameters $[n, n-um, 3]$.
  \item Let  $\gcd(\frac{q-1}r, n)=1$, $e>2$ and $p^j\equiv -1 \pmod {e}$ for a positive integer $j$. Assume that $j$ is the smallest positive integer such that $p^j\equiv -1 \pmod {e}$. Define $\gamma=\frac{s m}{2j}$.
  	\begin{enumerate}
  	 \item If $\gamma$ is odd, then $\C$ is an $\left[n, um, \frac{q^{m-1}-(e-1)q^{\frac{m-2}2}}{e} \right]$ code over $\gf(q)$ with  weight enumerator  
  	 $$\left[1+\left(\frac{q^m-1}{e}\right)z^{\frac{q^{m-1}-(e-1)q^{\frac{m-2}2}}{e}} +\left(q^m-1-\frac{q^m-1}e\right)z^{\frac{q^{m-1}+q^{\frac{m-2}2}}{e}}\right]^u.$$ 
  	 \item If $\gamma$ is even, then $\C$ is an $\left[n, um, \frac{q^{m-1}-q^{\frac{m-2}2}}{e} \right]$ code over $\gf(q)$ with weight enumerator  
  	 $$\left[1+\left(\frac{q^m-1}{e}\right)z^{\frac{q^{m-1}+(e-1)q^{\frac{m-2}2}}{e}} +\left(q^m-1-\frac{q^m-1}{e}\right)z^{\frac{q^{m-1}-q^{\frac{m-2}2}}{e}}\right]^u.$$
  	 \end{enumerate}
  	 The dual code $\C^{\bot}$ has parameters $[n, n-um, d^{\bot}]$, where 
\begin{align*}
     	d^{\bot}=\begin{cases}
	4 &{\rm if} ~m=4~{\rm and}~e=q+1,\\
	3 &{\rm otherwise}.
\end{cases}
     \end{align*}
  \end{enumerate}
  \end{theorem}
  
  \begin{proof} 
  In each case, we assume that $\gcd((q-1)/r, n)=1$. Hence, in each case we have  
  $\kappa=u$, and
  $$ \frac{q^m-1}{\gcd(q^m-1,(q-1)n)}=e.$$
  By Theorem \ref{thm5}, $\C$ has weight enumerator $W(z)^u$ if and only if $\Exp_1(\C)$ has weight enumerator $W(z^r)$. 
  Then the desired weight enumerator of $\C$ in each case follows from the weight enumerator of the code  $\Exp_1(\C)$ 
  given in Lemma \ref{lem8}. 
  
  We now settle the parameters of the dual code $\C^\perp$. The dimension of $\C^\perp$ follows from that of $\C$. 
  It remains to treat the minimum distance of $\C^\perp$. 
  Note that $\Exp_2(\C)$ is an irreducible $\lambda$-constacyclic code of length $\frac{q^m-1}{(q-1)e}$ over $\gf(q)$. If $e\leq q-1$, from Theorem \ref{thm7},  $\Exp_2(\C)$ is a $[\frac{q^m-1}{(q-1)e}, \frac{q^m-1}{(q-1)e}-m,3]$ code over $\gf(q)$. By Theorem \ref{thm6}, $\C^{\bot}$ is an $[n, n-um, 3]$ code over $\gf(q)$.
  
  If $e>2$ and $p^j\equiv -1 \pmod {e}$ for a positive integer $j$. From Lemma \ref{lem8}, the code $\Exp_2(\C)$ has weight enumerator
  $$1+\left(\frac{q^m-1}{e}\right)z^{\frac{q^{m-1}+(-1)^\gamma(e-1)q^{\frac{m-2}2}}{e}} +\left(q^m-1-\frac{q^m-1}{e}\right)z^{\frac{q^{m-1}-(-1)^{\gamma}q^{\frac{m-2}2}}{e}}.$$
 If $m=2$, then $\Exp_2(\C)^{\bot}$ is a $\left[\frac{q+1}{e}, \frac{q+1}{e}-2, 3\right]$ MDS code over $\gf(q)$. If $m\geq 4$, from the Pless power momnets, we have
  	 $$6A_3^{\bot}= \frac{(q^m-1)[q^m-(-1)^{\gamma}(e^2-3e+2)q^{\frac{m}2}-(q-2)e^2-3e+1]}{e^3}.$$ 
  	 It follows that $A_3^{\bot}=0$ if and only if 
  	 $$\Delta:=q^m-(-1)^{\gamma}(e^2-3e+2)q^{\frac{m}2}-(q-2)e^2-3e+1=0.$$
  	 \begin{itemize}
  	 	\item When $\gamma$ is odd, $q^{\frac{m}2}\equiv -1 \pmod {e}$. It is easy to check that 
  	  $$\Delta=e\left[ (q^{\frac{m}2}-q+2)e-3q^{\frac{m}2}-3 \right]+q^{m}+2q^{\frac{m}2}+1>0.$$
  	  Hence, $\Exp_2(\C)^{\bot}$ is a $\left[\frac{q^m-1}{(q-1)e},\frac{q^m-1}{(q-1)e}-m, 3\right]$ code over $\gf(q)$.
  	   \item When $\gamma$ is even, $q^{\frac{m}2}\equiv 1 \pmod {e}$. It is easy to check that 
  	  $$\Delta=(q^{\frac{m}2}-1)(q^{\frac{m}2}-1-e^2+3e)-(q-1)e^2=0$$
  	  if and only if 
     $$\left(\frac{q^{\frac{m}2}-1}{e}\right)\left[\left(\frac{q^{\frac{m}2}-1}{e}\right)-e+3 \right]=q-1.$$
     This equation holds if and only if $m=4$ and $e=q+1$. Therefore, $\Exp_2(\C)$ is a $$\left[\frac{q^m-1}{(q-1)e},\frac{q^m-1}{(q-1)e}-m, d\right]$$ code over $\gf(q)$, where
     \begin{align*}
     	d=\begin{cases}
	4 &{\rm if} ~m=4,~{\rm and}~e=q+1\\
	3 &{\rm otherwise}.
\end{cases}
     \end{align*}
  	 \end{itemize}
  	 According to Theorem \ref{thm6}, $d^{\bot}$ is equal to the minimum distance of $\Exp_2(\C)^{\bot}$. The desired result follows. 
  \end{proof}
  
  The results in the special case $u=1$ in Theorem \ref{thm9} may be proved with some results in \cite{SR2018}.  
  If this is possible, it may take some work to do so. 
  When $u>1$, the results in Theorem \ref{thm9} may not be derived from \cite{SR2018}. We now elaborate on 
  this statement. In \cite{SR2018}, the following parameter $L$ is used to determine the number 
  of nonzero weights in the irreducible $\lambda$-constacyclic code $\C$ in Eq. (\ref{eq1}): 
  \begin{eqnarray}\label{eqn-indianL}
  L:=\gcd\left( \frac{q^{\kappa \ell}-1}{q-1},  \frac{q^{\kappa \ell}-1}{nr}  \right). 
\end{eqnarray}    
The parameter $e$ defined in Eq. (\ref{eqn-oure}) may be different from the $L$ defined above. 
In \cite{SR2018}, Gaussian periods of order $L$ are used to express the weight distribution of 
the code $\C$ in Eq. (\ref{eq1}). This way of determining the weight distribution of $\C$ is infeasible in most cases, as $L$ could be very large and Gaussian periods of order $L$ are not evaluated for most orders $L$. Below is such example. 

\begin{example}\label{exam-comparison} 
Let $q=7$, $m=4$, $r=6$ and $n=800$. Then $\ord_{rn}(q)=8$. Let $\alpha$ be a primitive element of $\gf(q^8)$ with $\alpha^8 + 4\alpha^3 + 6\alpha^2 + 2\alpha + 3=0$. Then the corresponding $\lambda=3$, which is a primitive element of $\gf(7)$. Let $\C$ be the corresponding irreducible $\lambda$-constacyclic code in Eq. (\ref{eq1}). Then $\C$ has parameters $[800,8,343]$ and weight enumerator $1+4800z^{343}+5760000z^{686}=(1+2400z^{343})^2$. For this code, the parameter $e$ defined in Eq. (\ref{eqn-oure}) is $1$ and the corresponding $u=2$ in terms of the notation of Theorem \ref{thm9}. Hence, the parameters and the weight enumerator of this code $\C$ follow from the conclusions of the first case in Theorem \ref{thm9}. However, the parameter $L$ in Eq. (\ref{eqn-indianL}) is equal to $1201$, and  the parameters and the weight enumerator of this code $\C$ cannot be deduced from  the results in \cite{SR2018}, as Gaussian periods of order $1201$ over $\gf(7^8)$ are not evaluated. Note that $1201$ is a prime.   
\end{example} 

Below we present a number of examples for illustrating the cases in  Theorem \ref{thm9}. 

\begin{example} 
Let $m=2$, $q=2^4$, $r=q-1=15$, $e=1$, $u=3$. Then $n=u \frac{q^m-1}{(q-1)e}=51$. Let $\alpha$ be a generator of $\gf(q^6)^*$ with 
$  
\alpha^{24} + \alpha^{16} + \alpha^{15} + \alpha^{14} + \alpha^{13} + \alpha^{10} + \alpha^9 + \alpha^7 + \alpha^5 +
    \alpha^3 + 1=0. 
    $
Then the $\lambda$-constacyclic code $\C$ in Eq. (\ref{eq1}) has parameters $[51, 6, 16]$ and weight enumerator 
$$
16581375z^{48} + 195075z^{32} + 765z^{16} + 1. 
$$
$\C^\perp$ has parameters $[51, 45, 3]$. These results are consistent with the conclusions in the first case in Theorem \ref{thm9}. 
\end{example}

\begin{example} 
Let $m=2$, $q=7$, $r=q-1=6$, $e=2$, $u=3$. Then $n=u \frac{q^m-1}{(q-1)e}=12$. Let $\alpha$ be a generator of $\gf(q^6)^*$ with 
$ 
\alpha^6 + \alpha^4 + 5\alpha^3 + 4\alpha^2 + 6\alpha + 3=0
    $. 
Then the  $\lambda$-constacyclic code $\C$ in Eq. (\ref{eq1}) has parameters $[12, 6, 3]$ and weight enumerator 
$
(24z^4 + 24z^3 + 1)^3. 
$
$\C^\perp$ has parameters $[12, 6, 3]$. These results are consistent with the conclusions in the second case in Theorem \ref{thm9}. 
\end{example}    

\begin{example} 
Let $m=3$, $q=7$, $r=q-1=6$, $e=3$, $u=2$. Then $n=u \frac{q^m-1}{(q-1)e}=38$. Let $\alpha$ be a generator of $\gf(q^6)^*$ with 
$ 
\alpha^6 + \alpha^4 + 5\alpha^3 + 4\alpha^2 + 6\alpha + 3=0
    $. 
Then the $\lambda$-constacyclic code $\C$ in Eq. (\ref{eq1}) has parameters $[38, 6, 15]$ and weight enumerator 
$
(114z^{18} + 114z^{16} + 114z^{15} + 1 )^2. 
$
$\C^\perp$ has parameters $[38, 32, 3]$. These results are consistent with the conclusions in the third case in Theorem \ref{thm9}. 
\end{example}

\begin{example} 
Let $m=4$, $q=13$, $r=q-1=12$, $e=4$, $u=1$. Then $n=u \frac{q^m-1}{(q-1)e}=595$. Let $\alpha$ be a generator of $\gf(q^4)^*$ with 
$ 
\alpha^4 + 3\alpha^2 + 12\alpha + 2=0
    $. 
Then the $\lambda$-constacyclic code $\C$ in Eq. (\ref{eq1}) has parameters $[595, 4, 540]$ and weight enumerator 
$
7140z^{555} + 7140z^{552} + 7140z^{550} + 7140z^{540} + 1. 
$
$\C^\perp$ has parameters $[595, 591, 3]$. These results are consistent with the conclusions in the fourth case in Theorem \ref{thm9}. 
\end{example} 

\begin{example} 
Let $m=2$, $q=p=11$, $r=q-1=10$, $e=4$, $u=3$. Then $n=u \frac{q^m-1}{(q-1)e}=9$. Let $\alpha$ be a generator of $\gf(q^6)^*$ with 
$ 
\alpha^6 + 3\alpha^4 + 4\alpha^3 + 6\alpha^2 + 7\alpha + 2=0
    $. 
Then the  $\lambda$-constacyclic code $\C$ in Eq. (\ref{eq1}) has parameters $[9, 6, 2]$ and weight enumerator 
$
(90z^3 + 30z^2 + 1)^3. 
$
$\C^\perp$ has parameters $[9, 3, 3]$. These results are consistent with the conclusions in Case 5.a in Theorem \ref{thm9}. 
\end{example} 

\begin{example} 
Let $m=4$, $q=p=11$, $r=q-1=10$, $e=3$, $u=2$. Then $n=u \frac{q^m-1}{(q-1)e}=976$. Let $\alpha$ be a generator of $\gf(q^8)^*$ with 
$ 
\alpha^8 + 7\alpha^4 + 7\alpha^3 + \alpha^2 + 7\alpha + 2=0
    $. 
Then the $\lambda$-constacyclic code $\C$ in Eq. (\ref{eq1}) has parameters $[976, 8, 440]$ and weight enumerator 
$
(4880z^{451} + 9760z^{440} + 1)^2. 
$
$\C^\perp$ has parameters $[976, 968, 3]$. These results are consistent with the conclusions in Case 5.b in Theorem \ref{thm9}. 
\end{example}

The following is a list of corollaries of Theorem \ref{thm9}.  

\begin{corollary}\label{cor10}
	Let $m\geq 2$ and $q$ be a prime power. Let $r$ be a positive divisor of $q-1$ such that $\gcd(\frac{q-1}r,m)=1$. 
	Let $n=\frac{q^m-1}{q-1}$. Then the irreducible $\lambda$-constacyclic code over $\gf(q)$ with $\ord(\lambda)=r$ 
	in Eq. (\ref{eq1}) is monomially-equivalent to the Simplex code {\rm Simplex}($q, m$) and its dual is monomially-equivalent 
	to the Hamming code ${\rm Hamming}($q, m$)$. 
\end{corollary}

\begin{proof} 
It is easily seen that 
$$
\gcd\left(\frac{q-1}{r}, n\right)=\gcd\left(\frac{q-1}{r}, m\right)=1. 
$$
Let $u=1$ and $e=1$  in Theorem \ref{thm9}. We then deduce  from the first case in Theorem \ref{thm9} that $\C$ has the same parameters as the Simplex code and $\C^\perp$ has the same parameters as the Hamming code. It is well known that every linear code sharing the parameters of the Hamming code must be monomially-equivalent to the Hamming code \cite{HP2003}.    
\end{proof}

When $r=q-1$, Corollary \ref{cor10} was proven in \cite{FWF2017,HD2019}. With Corollary \ref{cor10}, we proved that there are more classes of constacyclic codes which are monomially-equivalent to the Hamming code.

\begin{corollary}\label{cor11-1} 
Let $n=2(\frac{q^m-1}{q-1})$, where $m\geq 2$ and $q$ is an odd prime power. Let $r$ be a positive divisor of $q-1$ such that $\gcd(\frac{q-1}r,2m)=1$. Let $\C$ be the irreducible $\lambda$-constacyclic code in Eq. (\ref{eq1}). Then $\C$ is an $[n,2m, q^{m-1}]$ two-weight code with weight enumerator $$1+2(q^m-1)z^{q^{m-1}}+(q^m-1)^2z^{2q^{m-1}}.$$ 
Its dual code has parameters $[n, n-2m, 3]$. 
	\end{corollary} 
	
\begin{proof} 
It is straightforward to see that 
$$
\gcd\left(\frac{q-1}{r}, n  \right)= \gcd\left(\frac{q-1}{r}, 2m   \right)=1. 
$$
Let $u=2$ and $e=1$  in Theorem \ref{thm9}. Then the desired conclusions follow.  
\end{proof}

 \begin{example}
 	Let $m=3$, $q=3$ and $r=2$. Let $\alpha$ be a generator of $\gf(q^6)^*$ with $\alpha^6 + 2\alpha^4 + \alpha^2 + 2\alpha + 2=0$. 
 	Then the irreducible negacyclic code $\C$ of length $q^m-1$ over $\gf(q)$ in Eq. (\ref{eq1}) has parameters $[26,6,9]$ and weight enumerator $1+52z^{9}+676z^{18}$. The dual code $\C^{\bot}$ has parameters $[26,20,3]$. 
 \end{example}

\begin{corollary}\label{cor12-1} 
Let $n=3(\frac{q^m-1}{q-1})$, where $m\geq 2$ and $q \equiv 1 \pmod{3}$ is a prime power. Let $r$ be a positive divisor of $q-1$ such that $\gcd(\frac{q-1}r,3m)=1$. Let $\C$ be the irreducible $\lambda$-constacyclic code in Eq. (\ref{eq1}). Then $\C$ is an $[n,3m, q^{m-1}]$ three-weight code with weight enumerator $$1+3(q^m-1)z^{q^{m-1}}+3(q^m-1)^2z^{2q^{m-1}}+(q^m-1)^3z^{3q^{m-1}},$$ and $\C^\perp$ has parameters $[n, n-3m, 3]$.  	
\end{corollary} 

\begin{proof}
Let $q=rt+1$. Then 
\begin{eqnarray*}
\gcd\left(\frac{q-1}{r}, n\right) 
&=& \gcd\left(\frac{q-1}{r}, 3\frac{q^m-1}{q-1}\right) \\
&=& \gcd(t, 3(q^{m-1}+q^{m-2}+ \cdots + q +1)) \\ 
&=& \gcd\left(\frac{q-1}{r}, 3m\right) \\
&=& 1.  
\end{eqnarray*} 
Let $u=3$ and $e=1$ in Theorem \ref{thm9}. Then the desired conclusions follow. 
\end{proof}

 \begin{example}
 	Let $m=4$, $q=4$ and $r=q-1$. Let $\alpha$ be a generator of $\gf(q^{12})^*$ with
 	$
 	\alpha^{24} + \alpha^{16} + \alpha^{15} + \alpha^{14} + \alpha^{13} + \alpha^{10} + \alpha^9 + \alpha^7 + \alpha^5 +
    \alpha^3 + 1=0. 
 	$ 
 	Then the irreducible $\lambda$-constacyclic code $\C$ of length $3(\frac{q^m-1}{q-1})$ over $\gf(q)$ in Eq. (\ref{eq1}) has parameters $[255,12,64]$ and weight enumerator 
 	$$1+765 z^{64}+195075z^{128}+16581375 z^{192}.$$ 
 	The dual code $\C^{\bot}$ has parameters $[255,243,3]$. 
 \end{example}

\section{A family of irreducible cyclic codes over $\gf(q)$ with parameters $[2n, (n-1)/2, d \geq 2(\sqrt{n}+1)]$}\label{sec-ding1} 

 Throughout this section, let $q$ be an odd prime power and let $n$ be an odd prime such that $m:=\ord_{2n}(q)=(n-1)/2$. 
 Note that $\gcd(2,n)=1$. We have then 
 $\kappa=1$. Consequently, $\ord_n(q)=\ord_{2n}(q)=(n-1)/2$. 
  
Let $\alpha$ be a primitive element of $\gf(q^m)$ and put $\beta=\alpha^{(q^m-1)/2n}$. Then $\beta^n=-1$. By definition and assumption, we 
have 
$$
C_1^{(q,2n)}=\{1, q, q^2, \ldots, q^{m-1}\} \bmod{2n}. 
$$  
Clearly, all the $(n-1)/2$ elements in $C_1^{(q,2n)}$ are odd. Take any $h$ from $\{2i+1: 0 \leq i \leq n-1\} \setminus 
(C_1^{(q,2n)} \cup \{n\})$. Then the cosets $C_h^{(q,2n)}$ and $C_1^{(q,2n)}$ are disjoint and 
$$
C_1^{(q,2n)} \cup C_h^{(q,2n)} =  \{2i+1: 0 \leq i \leq n-1\} \setminus \{n\}.  
$$ 
Then for $i \in \{1, h\}$ the minimal polynomial  
$$
\m_{\beta^i}(x) =\prod_{j \in C_i^{(q,2n)}} (x-\beta^j)
$$  
and 
$$
x^n+1=(x+1) \m_{\beta}(x)  \m_{\beta^h}(x).  
$$
Let $\C(q,n,i; 2)$ denote the negacyclic code of length $n$ over $\gf(q)$ with check polynomial $ \m_{\beta^i}(x)$. Then 
the dimension of  $\C(q,n,i; 2)$ equals $(n-1)/2$ for $i \in \{1, h\}$. Consequently, its dual has dimension  $(n+1)/2$. We first prove the following theorem. 

\begin{theorem}\label{thm-mainfeb26}
Let notation and assumptions be the same as before. Then  the irreducible negacyclic code $\C(q,n,1; 2)$ has parameters  $[n, (n-1)/2, d]$  with 
$d \geq \sqrt{n}+1$ and its dual $\C(q,n,1; 2)^\perp$ has parameters  $[n, (n+1)/2, d^\perp]$ with 
$d^\perp \geq \sqrt{n}$. 
\end{theorem} 

For the code $\C(q,n,i; 2)$, the lower bound on the minimum distance developed in Theorem 4.3 in \cite{SR2018} is negative and not useful. We have to prove the desired square-root bounds with a different approach.  To this end, we need to do some preparations.  

\begin{lemma}\label{lem-2801}
Let $q$ be an odd prime power and let $n$ be an odd prime such that $\ord_{2n}(q)=(n-1)/2$ and $n >q$. Then $q$ is a quadratic residue modulo $n$.  
\end{lemma}

\begin{proof}
It was shown at the beginning of this section that  $\ord_n(q)=\ord_{2n}(q)=(n-1)/2$. Let $\zeta$ be a primitive root of $n$. Since $n>q$, then there exists an integer $i$ with $1 \leq i <n-1$ such that $q=\zeta^i \bmod{n}$. It is well known that $\ord_n(q)=(n-1)/\gcd(i, n-1)$. As a result, $\gcd(i, n-1)=2$ and $i$ must be even. The desired result then follows. 
\end{proof}

\begin{lemma}\label{lem-RQCbound}
 $\Exp_3(\C(q,n,1; 2))$ has parameters $[n, (n-1)/2, d \geq \sqrt{n}+1]$ and $\Exp_3(\C(q,n,1; 2))^\perp$ has parameters $[n, (n+1)/2, d^\perp \geq \sqrt{n}].$
\end{lemma} 

\begin{proof}
Note that $\beta^2$ is an $n$-th primitive root of unity in $\gf(q^m)$. Let $\QR(n)$ and $\QN(n)$ denote the set of quadratic residues and nonresidues modulo $n$, respectively. It follows from Lemma \ref{lem-2801} that the canonical factorization of $x^n-1$ over $\gf(q)$ is given by   
\begin{eqnarray}
x^n-1= (x-1) \left( \prod_{i \in \QR(n)} (x-\beta^{2i}) \right)   \left( \prod_{i \in \QN(n)} (x-\beta^{2i}) \right).  
\end{eqnarray} 
Since $\m_{\beta^2}(x)$ is an irreducible divisor of $x^n-1$ and $\beta^2 \not\in \gf(q)$, we have 
$$
\m_{\beta^2}(x)= \prod_{i \in \QR(n)} (x-\beta^{2i}) 
$$ 
or 
$$
\m_{\beta^2}(x)= \prod_{i \in \QN(n)} (x-\beta^{2i}).  
$$ 
Consider now the code $\Exp_3(\C(q,n,1; 2))$, which is the third code associated to the irreducible negacyclic code $\C(q,n,1; 2)$ and has check polynomial $\m_{\beta^2}(x)$. The check polynomial $\m_{\beta^2}(x)$ above shows that  $\Exp_3(\C(q,n,1; 2))$ is a quadratic-residue code. The desired conclusions then follow from \cite[Theorem 6.6.22]{HP2003}. 
\end{proof}

Since $\gcd(2, n)=1$, the desired conclusions of Theorem \ref{thm-mainfeb26} then follow from Theorem \ref{thm-bridge1} and Lemma \ref{lem-RQCbound}. This completes the proof of Theorem \ref{thm-mainfeb26}. Theorem \ref{thm-bridge1} tells us that $\C(q,n,1; 2)$ and $\Exp_3(\C(q,n,1; 2))$ are permutation-equivalent. Hence, the negacyclic code $\C(q,n,1; 2)$ is a counterpart of the corresponding quadratic-residue code $\Exp_3(\C(q,n,1; 2))$. Combining Theorems \ref{thm5} and \ref{thm-mainfeb26}, we obtain the following main result of this section.  

\begin{theorem}\label{thm-june161}
$\Exp_1(\C(q,n,1; 2))$ is an irreducible cyclic code with parameters $[2n, (n-1)/2, d \geq 2(\sqrt{n}+1)]$. 
\end{theorem} 

\begin{example} 
The irreducible cyclic code $\Exp_1(\C(3,11,1; 2))$ has parameters $[22, 5, 12]$ and is distance-optimal. 
\end{example} 

One question is whether there are infinitely many primes $n$ such that $\ord_{2n}=(n-1)/2$ for a fixed odd prime power $q$. This question may be open. But experimental data indicates that the answer to this question is positive. Table \ref{tab-ding1} contains a list of examples of the code $\C(q,n,1; 2)$, where only the parameters $q$, $n$, $d(\C(q,n,1; 2))$, $d(\C(q,n,1; 2)^\perp)$ are listed, and the first best distance and the second best distance denote the minimum distance of the best linear code of length $n$ and dimension $(n-1)/2$ and the best linear code of length $n$ and dimension $(n+1)/2$ over $\gf(q)$, respectively, maintained at http://www.codetables.de/. The experimental data in Table \ref{tab-ding1} shows that the codes $\C(q,n,1; 2)$ and $\C(q,n,1; 2)^\perp$ are optimal or the best known in every case except one.

\begin{table}[ht]
\begin{center}
\caption{The code $\C(q,n,1; 2)$ and its dual}\label{tab-ding1}
\begin{tabular}{rrrrrr} \hline
$q$  &  $n$  & $d$ & best distance & $d^\perp$ &  best distance  \\ \hline
$3$  &   $11$     &  $6$   &  $6$, optimal              &  $5$             &  $5$, optimal \\  \hline 
$3$  &   $23$      &  $9$  &  $9$,  optimal              &  $8$             &  $8$, optimal \\  \hline 
$3$  &   $37$     &  $11$   &  $12$             &  $10$             &  $10$ \\  \hline 
$3$  &   $47$      &  $15$  &   $15$             &  $14$             &  $14$ \\  \hline 
$3$  &   $59$     &  $18$   &  $18$              &  $17$             &  $17$ \\  \hline 
$3$    & $71$    &  $$ 18 &   $18$              &  $17$             &  $17$ \\  \hline 
$5$   & $11$    &  $6$  &   $6$, optimal              &  $5$             &  $5$, optimal \\  \hline 
$5$    & $19$    &  $8$  &   $8$              &  $7$             &  $7$ \\  \hline 
$5$   & $29$    &  $12$  &   $12$, optimal              &  $11$             &  $11$, optimal  \\  \hline 
$5$    & $41$    &  $14$  &   $14$              &  $13$             &  $13$ \\  \hline 
$7$   & $31$    &  $13$  &   $13$             &  $12$             &  $12$ \\  \hline 
$7$    & $47$    &  $17$  &   $17$              &  $16$             &  $16$ \\  \hline 
\end{tabular}
\end{center}
\end{table}

\section{A family of irreducible cyclic codes over $\gf(q)$ with parameters $[(q-1)n, (n-1)/2, d \geq (q-1)(\sqrt{n}+1)]$}\label{sec-ding2} 

Throughout this section, let $q>2$ be a prime power and let $n>q$ be an odd prime such that $m:=\ord_{(q-1)n}(q)=(n-1)/2$. Note that $q<n$ and $n$ is a prime. We have $\gcd(q-1,n)=1$ and $\kappa=1$. Consequently, $\ord_n(q)=\ord_{(q-1)n}(q)=(n-1)/2$.

Let $\alpha$ be a primitive element of $\gf(q^m)$ and let $\beta=\alpha^{(q^m-1)/(q-1)n}$. 
Put $\lambda=\beta^n$. Then $\lambda$ is a primitive element of $\gf(q)$. Let $h$ denote the multiplicative inverse of $n$ modulo $q-1$, where $1 \leq h \leq q-2$. Then we have  
$$
C_{h n}^{(q, (q-1)n)}=\{h n\}
$$ 
and 
$$
C_1^{(q,(q-1)n)}=\{1, q, q^2, \ldots, q^{m-1}\} \bmod{(q-1)n}.  
$$ 
Take any $t$ from $\{(q-1)i+1: 0 \leq i \leq n-1\} \setminus 
(C_1^{(q,(q-1)n)} \cup C_{h n}^{(q, (q-1)n)})$. It can be verified that the cosets $C_1^{(q,(q-1)n)}$, $C_t^{(q,(q-1)n)}$ and $C_{h n}^{(q,(q-1)n)}$ form a partition of the set 
$
 \{(q-1)i+1: 0 \leq i \leq n-1\}.  
$ 
Then for $i \in \{1, t\}$ the minimal polynomial  
$$
\m_{\beta^i}(x) =\prod_{j \in C_i^{(q,(q-1)n)}} (x-\beta^j). 
$$  
Let $\C(q,n,i)$ denote the $\lambda$-constacyclic code of length $n$ over $\gf(q)$ with check polynomial $ \m_{\beta^i}(x)$ for $i \in \{1, t\}$. Then the dimension of  $\C(q,n,i)$ equals $(n-1)/2$ for $i \in \{1, t\}$. Consequently, its dual has dimension  $(n+1)/2$. We now prove the following theorem.  

\begin{theorem}\label{thm-mainfeb27}
Let notation and assumptions be the same as before. Then the irreducible constacyclic code $\C(q,n,1)$ has parameters $[n, (n-1)/2, d]$ with $d \geq \sqrt{n}+1$ and $\C(q,n,1)^\perp$ has parameters $[n, (n+1)/2, d^\perp]$ with $d^\perp \geq \sqrt{n}$. 
\end{theorem}

To prove Theorem \ref{thm-mainfeb27}, we need to do some preparations below. We first prove the following lemma. 

\begin{lemma}\label{lem-2802}
Let $q$ be a prime power and let $n>q$ be an odd prime  such that $\ord_{(q-1)n}(q)=(n-1)/2$. Then $q$ is a quadratic residue modulo $n$.  
\end{lemma}

\begin{proof}
It was shown at the very beginning of this section that  $\ord_n(q) =\ord_{(q-1)n}(q)=(n-1)/2$. Let $\zeta$ be a primitive root of $n$. Since $n>q$, then there exists an integer $i$ with $1 \leq i <n-1$ such that $q=\zeta^i \bmod{n}$. It is well known that $\ord_n(q)=(n-1)/\gcd(i, n-1)$. As a result, $\gcd(i, n-1)=2$ and $i$ must be even. The desired result then follows. 
\end{proof}  

\begin{lemma}\label{lem-RQCboundq1}
 $\Exp_3(\C(q,n,1))$ has parameters $[n, (n-1)/2, d \geq \sqrt{n}+1]$ and 
 $\Exp_3(\C(q,n,1))^\perp$ has parameters $[n, (n+1)/2, d^\perp \geq \sqrt{n}]$. 
\end{lemma} 

\begin{proof}
Let $\QR(n)$ and $\QN(n)$ denote the set of quadratic residues and nonresidues modulo $n$, respectively. Note that $\beta^{q-1}$ is an $n$-th primitive root of unity in $\gf(q^m)$. It follows from Lemma \ref{lem-2802} that the canonical factorization of $x^n-1$ over $\gf(q)$ is given by   
\begin{eqnarray}
x^n-1= (x-1) \left( \prod_{i \in \QR(n)} (x-\beta^{(q-1)i}) \right)   \left( \prod_{i \in \QN(n)} (x-\beta^{(q-1)i}) \right).  
\end{eqnarray} 
Since $\m_{\beta^{q-1}}(x)$ is an irreducible divisor of $x^n-1$ and $\beta^{q-1} \not\in \gf(q)$, 
we have 
$$
\m_{\beta^{q-1}}(x)= \prod_{i \in \QR(n)} (x-\beta^{(q-1)i}) 
$$ 
or 
$$
\m_{\beta^{q-1}}(x)= \prod_{i \in \QN(n)} (x-\beta^{(q-1)i}).  
$$ 
Consider now the code $\Exp_3(\C(q,n,1))$, which is the third code associated to the irreducible constacyclic code $\C(q,n,1)$ and has check polynomial $\m_{\beta^{q-1}}(x)$. The check polynomial $\m_{\beta^{q-1}}(x)$ above shows that $\Exp_3(\C(q,n,1))$ is a quadratic-residue code. The desired conclusions then follow from \cite[Theorem 6.6.22]{HP2003}. 
\end{proof}
 
By assumption $n>q$ and $n$ is a prime. Hence, $\gcd(q-1, n)=1$. The desired conclusions of Theorem \ref{thm-mainfeb27} then follow from Theorem \ref{thm-bridge1} and Lemma \ref{lem-RQCboundq1}. Theorem \ref{thm-bridge1} tells us that $\C(q,n,1)$ and $\Exp_3(\C(q,n,1))$ are permutation-equivalent. Hence, the constacyclic code $\C(q,n,1)$ is a counterpart of the corresponding quadratic-residue code $\Exp_3(\C(q,n,1))$. Combining Theorems  \ref{thm5} and \ref{thm-mainfeb27}, we get the following main result of this section.  

\begin{theorem}\label{thm-june162} 
The set $\Exp_1(\C(q,n,1))$ is an irreducible cyclic code over $\gf(q)$ with parameters $[(q-1)n, (n-1)/2, d \geq (q-1)(\sqrt{n}+1)]$. 
\end{theorem} 

\begin{example} 
The irreducible cyclic code $\Exp_1(\C(3,11,1))$ has parameters $[22, 5, 12]$ and is distance-optimal. 
\end{example} 

Similarly, one would ask if there are infinitely many primes $n$ such that $\ord_{(q-1)n}=(n-1)/2$ for a fixed prime power $q>2$. This question may be open. But experimental data indicates that the answer to this question is positive. Table \ref{tab-ding2} contains a list of examples of the code $\C(q,n,1)$, where only the parameters $q$, $n$, $d(\C(q,n,1))$, $d(\C(q,n,1)^\perp)$ are listed, and the first best distance and the second best distance denote the minimum distance of the best linear code of length $n$ and dimension $(n-1)/2$ and the best linear code of length $n$ and dimension $(n+1)/2$ over $\gf(q)$, respectively, maintained at http://www.codetables.de/. The experimental data in Table \ref{tab-ding2}  shows that the codes $\C(q,n,1)$ and $\C(q,n,1)^\perp$ are optimal or the best known in every case except three cases. 

 \begin{table}[ht]
\begin{center}
\caption{The code $\C(q,n,1)$ and its dual}\label{tab-ding2}
\begin{tabular}{rrrrrr} \hline
$q$  &  $n$  & $d$ & best distance & $d^\perp$ &  best distance  \\ \hline
$4$    & $7$    &  $4$  &   optimal              &  $3$             &  optimal \\  \hline 
$4$   & $11$    &  $6$  &   optimal              &  $5$             &  optimal \\  \hline 
$4$    & $13$    &  $6$  &   optimal              &  $5$             &  optimal \\  \hline 
$4$   & $19$    &  $8$  &   optimal              &  $7$             &  optimal \\  \hline 
$4$    & $23$    &  $8$  &   $9$              &  $7$             &  $8$ \\  \hline 
$4$   & $29$    &  $12$  &   optimal              &  $11$             &  optimal \\  \hline 
$4$    & $37$    &  $12$  &   $12$              &  $11$             &  optimal \\  \hline 
$4$   & $47$    &  $12$  &   $14$             &  $11$             &  $13$ \\  \hline 
$4$    & $59$    &  $15$  &   $17$              &  $14$             &  $16$ \\  \hline 
$4$   & $61$    &  $18$  &   $18$              &  $17$             &  $17$ \\  \hline 
$5$    & $11$    &  $6$  &   optimal              &  $5$             &  optimal \\  \hline 
$5$   & $19$    &  $8$  &   $8$              &  $7$             &  $7$ \\  \hline 
$5$    & $29$    &  $12$  &   optimal             &  $11$             &  optimal \\  \hline 
$5$   & $41$    &  $14$  &   $14$              &  $13$             &  $13$  \\  \hline 
$7$   & $31$    &  $13$  &   $13$              &  $12$             &  $12$ \\  \hline 
\end{tabular}
\end{center}
\end{table}

\section{The objective of studying the two families of constacyclic codes in Sections \ref{sec-ding1} and \ref{sec-ding2}} \label{sec:7}

Recall that the constacyclic codes $\C(q, n, 1; 2)$ studied in Section  \ref{sec-ding1} and  $\C(q, n, 1)$ treated in Section  \ref{sec-ding2} were proved to be permutation-equivalent to a quadratic-residue code under the condition $\ord_{2n}(q)=(n-1)/2$ 
 and $\ord_{(q-1)n}(q)=(n-1)/2$, respectively. Then one would question the objective of studying 
the constacyclic codes $\C(q, n, 1; 2)$ and $\C(q, n, 1)$. The major objective of doing this is to obtain the 
parameters of the two irreducible cyclic codes $\Exp_1(\C(q,n,1;2))$ documented in Theorem \ref{thm-june161}  and 
$\Exp_1(\C(q,n,1))$ documented in Theorem \ref{thm-june162}, where the constacyclic code  $\C(q, n, 1; 2)$ (or  $\C(q, n, 1)$)  serves as a bridge between the irrecucible cyclic codes $\Exp_1(\C(q,n,1;2))$ and $\Exp_3(\C(q,n,1;2))$ (or $\Exp_1(\C(q,n,1))$ and $\Exp_3(\C(q,n,1))$) of length $2n$ and $n$ (or $(q-1)n$  and $n$), respectively.  

Recall that we assocaited to  each irreducible $\lambda$-constacyclic code $\C$ of length  $n$ an irreducible cyclic code  
$\Exp_1(\C)$ of length $rn/\kappa$ and another  irreducible cyclic code $\Exp_3(\C)$ of length $n$. By definition, the two cyclic codes $\Exp_1(\C)$ and $\Exp_3(\C)$ are related to some extent in general via the bridge code $\C$, and are closely 
related in some special cases. One may be able to obtain results about $\Exp_1(\C)$ from those about $\Exp_3(\C)$ or the other way 
around in some special cases via the bridge constacyclic code $\C$.  Hence, studying the brdging constacyclic cocde $\C$ is 
a key step in this approach. There are other ways to associate a cyclic code to a constacycic codes \cite{CFLL12}.

\section{Summary and concluding remarks}\label{sec4}
   	 
 The main contributions of this paper are the following:
 \begin{enumerate}
 \item The relations among the weight distributions of the irreducible constacyclic code $\C$ and the related code 
 $\Exp_1(\C)$ and the irreducible constacyclic code $\Exp_2(\C)$  were discovered (see Theorem \ref{thm5}), and the relations among the three codes were found (see Theorem \ref{thm4}). 
  \item For the irreducible constacyclic code $\C$ in Eq. (\ref{eq1}), another irreducible cyclic code $\Exp_3(\C)$ of the same length was assocaoted. Relations 
 between the two codes in a special case were established in Theorem \ref{thm-bridge1}, which serves as a bridge between 
 irreducible constacyclic codes and irreducible cyclic codes.  
 \item The weight distributions of several families of irreducible constacyclic codes were settled. Several families of constacyclic codes with a few weights were produced. These results were documented in Theorem \ref{thm9} and its corollaries. 
 The dual codes of these irreducible constacyclic codes were also studied (see Theorem \ref{thm9} and its corollaries). Several families of constacyclic codes with optimal parameters were presented.  
  \item A family of irreducible cyclic codes over $\gf(q)$ with parameters $[2n, (n-1)/2, d \geq 2(\sqrt{n}+1)]$ was constructed in Section 
 \ref{sec-ding1}. 
   \item A family of irreducible cyclic codes over $\gf(q)$ with parameters $[(q-1)n, (n-1)/2, d \geq (q-1)(\sqrt{n}+1)]$ was constructed in Section 
 \ref{sec-ding2}. 
 \end{enumerate} 
 
It was shown in \cite[Theorem 3.6]{SR2018} that the weight distribution of the irreducible constacyclic code $\C$ in Eq. 
(\ref{eq1}) can be expressed in terms of the Gaussian periods of order $L$, where $L$ was defined in Eq. (\ref{eqn-indianL}).  However, 
Gaussian periods were evaluated only in a few cases. Consequently, the weigh distribution of the irreducible constacyclic code 
$\C$ in Eq. (\ref{eq1}) was known only in a few cases \cite{SR2018}. As explained in Example \ref{exam-comparison}, 
Theorem \ref{thm9} complements the work in \cite{SR2018}.  Hence, Corollaries \ref{cor11-1} and \ref{cor12-1} of Theorem \ref{thm9} may not be easily derived from the results in 
\cite{SR2018}. 
This paper treated the  code $\C$ in Eq. (\ref{eq1}) with length $n$ of several special forms only, 
while reference \cite{SR2018} documented some general results of $\C$ with general length $n$. Hence, the focuses of this 
paper and \cite{SR2018} are different. It may be extremely difficult to use the results in \cite{SR2018} to study the codes 
and their duals presented in Sections \ref{sec-ding1} and \ref{sec-ding2}, as Gaussian periods of such large orders are not 
evaluated. Another difference between this paper and \cite{SR2018} is that the dual code $\C^\perp$ was studied in this paper, 
while the dual code of $\C$ was not touched in \cite{SR2018}.  

The evaluation of Gaussian periods is related to several areas of number theory and is known to be a hard problem. 
Hence, it is very hard to determine the weight distributions of irreducible cyclic codes \cite{dingyang2013} and irreducible constacyclic codes \cite{SR2018}. The reader is cordially invited to make progress on this topic. 
  
\section{Acknowledgements} 
The authors are very grateful to Prof. Anuradha Sharma for providing reference \cite{SR2018}. All the code examples in this paper 
were computed with the Magma software package.

\end{document}